\documentclass[12pt,reqno]{amsart}
\usepackage[abbrev]{amsrefs} 
\usepackage{amsmath,amsthm,amsfonts,amssymb,multicol}
\usepackage{etoolbox}
\usepackage{graphicx}
\usepackage{booktabs}
\usepackage{array}
\usepackage{fullpage}
\usepackage{subfigure}
\usepackage{enumerate}
\usepackage{url}

\usepackage[english]{babel}
\usepackage{verbatim} %multi-line comments
\usepackage{enumitem}
\usepackage{bbm} %math symbols
\usepackage[normalem]{ulem} %dashed & pointed underline
\usepackage{bm}
\usepackage{mathtools}
\usepackage[pdftex]{hyperref}
\usepackage{graphicx}
\usepackage{xcolor}
\usepackage{cancel}
\hypersetup{pdfpagemode=}

\newcommand{\abs}[1]{\ensuremath{\left\lvert#1\right\rvert}}

\newcommand{\del}[1]{\left(#1\right)}
\newcommand{\sbr}[1]{\left[#1\right]}
\newcommand{\Set}[1]{\left\{#1\right\}}

\newcommand{\bA}{\mathbb{A}}

\newcommand{\Mr}[1]{\mathrm{#1}}

\newcommand{\cD}{\mathcal{D}}
\newcommand{\cC}{\mathcal{C}}
\newcommand{\cE}{\mathcal{E}}
\newcommand{\cF}{\mathcal{F}}
\newcommand{\cT}{\mathcal{T}}

\newcommand{\nc}{\newcommand}

\nc{\h}{\delta}
\nc{\G}{\Gamma}
\nc{\et}{\eta} 
\nc{\gam}{\gamma}
\nc{\ka}{\kappa}
\nc{\lam}{\lambda}
\nc{\Lam}{\Lambda}
\nc{\ta}{\tau}
\nc{\w}{\omega}
\nc{\io}{\iota}
\nc{\s}{\sigma}
\nc{\vphi}{\varphi}

\nc{\e}{\epsilon}

\renewcommand{\k}{\kappa}

\newtheorem{thm}{Theorem}[section]

\newtheorem{prop}[thm]{Proposition}

\theoremstyle{definition}

\theoremstyle{remark}
\newtheorem{rmk}[thm]{Remark}

\newcommand{\thmref}[1]{Theorem~\ref{#1}}

\newcommand{\secref}[1]{Section~\ref{#1}}

\newcommand{\propref}[1]{Proposition~\ref{#1}}

\numberwithin{equation}{section}

\begin{document}
\title{
A ballistic upper bound on the accumulation of bosonic on-site energies
}

\author{Tomotaka Kuwahara}
\address{\textnormal{(Tomotaka Kuwahara)} Analytical quantum complexity RIKEN Hakubi Research Team,
RIKEN Center for Quantum Computing (RQC), Wako, Saitama 351-0198, Japan}
\email{tomotaka.kuwahara@riken.jp}

\author{Marius Lemm}
	\address{\textnormal{(Marius Lemm)} Department of Mathematics, University of T\"ubingen, 72076 T\"ubingen, Germany }
	\email{marius.lemm@uni-tuebingen.de}
	
	\author{Carla Rubiliani}
	\address{\textnormal{(Carla Rubiliani)}  Department of Mathematics, University of T\"ubingen, 72076 T\"ubingen, Germany }
	\email{carla.rubiliani@uni-tuebingen.de}
\date{\today}

\newcommand{\qmexp}[1]{\big\langle #1\big\rangle}
\newcommand{\A}{\mathbb A}

\maketitle

\begin{center}
     \textit{Dedicated to Israel Michael Sigal on the occasion of his 80th birthday.}
\end{center}

     \begin{abstract}

In this note, we study transport properties of the dynamics generated by translation-invariant and possibly long-ranged Hamiltonians of Bose-Hubbard type.
For translation-invariant initial states with controlled boson density, we improve the known bound on the local repulsive energy at time $t$ from $\langle n^2_x\rangle_t\lesssim t^{2d}$ to $\langle n^2_x\rangle_t\lesssim t^d$. This shows that bosonic on-site energies accumulate at most ballistically. Extending the result to higher moments would have powerful implications for bosonic Lieb-Robinson bounds.
While previous approaches focused on controlling particle transport, our proof develops novel ASTLOs (adiabatic space-time localization observables) that are able to track the growth of local boson-boson correlations. 
\end{abstract}

\section{Introduction}

In relativity, energy and information cannot travel faster than the speed of light. In other areas of physics, there are many instances of system-dependent speed limits that are much slower than the speed of light, such as the speed of sound.
This prompts a fundamental question: whether in quantum systems, even non-relativistic ones, similar speed limits emerge on the propagation of information or other physical quantities.
In a seminal work from 1972, Lieb and Robinson addressed this question, establishing bounds on the speed of quantum information propagation, for quantum spin systems which are governed by local and bounded interactions \cite{lieb1972finite}. 
In the early 2000s, Hastings and others employed and extended the Lieb-Robinson bounds (LRBs)  as a powerful analytical tool to answer many long-standing open problems in mathematical physics \cites{Hastings2004-ph,hastings2006spectral,nachtergaele2006lieb,hastings2007area}. 
Since then, LRBs have gained widespread recognition, largely due to their versatility and strength in rigorous proofs; see e.g., \cites{hastings2005quasiadiabatic, bravyi2006lieb,nachtergaele2006propagation} and the reviews \cites{gogolin2016equilibration,chen2023speed}.
As a natural consequence, the need to extend them to more general settings has grown into a research area of its own, e.g., the case of long-range bounded interactions was treated in \cites{CGS,Fossetal,TranEtAl1,TranEtAl2,TranEtal4, KS_he,lemm2025enhanced}.

 However, our understanding of  LRBs for unbounded interactions remains incomplete.  A natural and physically relevant source of unbounded interactions on lattices is the 
 Bose-Hubbard Hamiltonian 
\begin{align}\label{bh hamilt}
H_{\mathrm{B-H}}=\sum_{\substack{x,y\in\Lambda}}J_{x,y} b_x^\dagger b_y+U\sum_{x\in\Lambda} n_x(n_x-1)-\mu\sum_{x\in\Lambda}n_x.  
\end{align}
While closely related oscillator models could still be treated in certain perturbative regimes \cite{Nachtergaele2009-rl}, for the Bose-Hubbard Hamiltonian \eqref{bh hamilt} the prior techniques break down due to the possibility of boson accumulation and thus the problem of information propagation must be fundamentally reconsidered.
Starting from 2011 \cite{SHOE}, it became apparent that controlling particle propagation is a crucial input step to derive LRBs in the Bose-Hubbard setting, as it allows to quantitatively control the local accumulation of  bosons, which is the source of the unbounded interactions. This led to a focus on deriving particle propagation \cite{YL,Kuwahara2024-nf,MR4470244,LRSZ,SiZh,lemm2023microscopic,Lemm2025-ov} also for the related problem of macroscopic transport \cite{LRSZ,MR4419446,van2024optimal,faupin2025macroscopic}.
Using particle propagation bounds, ballistic LRBs have been investigated and proven for various kinds of special initial states \cite{SHOE,YL,MR4470244,LRSZ,SiZh}; see also \cite{kuwahara2021lieb,wang2020tightening}. For the more general class of bounded-density initial states,  \cite{KVS} proved a bound that allows for information to accelerate in $d>1$ --- the velocity bound scales as $v\sim t^{d-1}$. This has raised the open question whether such accelerated information propagation is indeed possible in Bose-Hubbard systems or if it is in fact prevented under more stringent assumptions, e.g., on the initial state.
In this direction, a recent work \cite{Kuwahara2024-nf} by two of us showed that energy conservation can meaningfully control information propagation further, showing that translation invariance and a strong local repulsion $n_x^p$, with $p$ sufficiently large, provide an almost-ballistic bound on information propagation for any bounded energy density  initial state.

In this work, we continue  the central theme of \cite{Kuwahara2024-nf} of tracking local energy accumulation. We prove a new bound on the accumulation of local repulsive energy under the assumption of translation invariance. We show
\begin{align}
     \qmexp{\psi(t),n_x^2\psi(t)}\lesssim  t^{d}.\label{improve sclaing}
 \end{align}
In terms of $t$-scaling, our result compares favorably to previous works \cites{KVS,lemm2023microscopic,LRSZ,Lemm2025-ov}, where the following scaling was obtained
 \begin{align}
     \qmexp{\psi(t),n_x^2\psi(t)}\lesssim  t^{2d}.
 \end{align}
 The underlying reason is that the prior works only relied on bounding particle transport ballistically. Indeed, a constant speed of particle transport allows for a scenario where, given a fixed site $x$, all $\sim t^d$ particles that are initially on the ball around $x$ of radius $\sim t$ may reach $x$ within time $t$, which would indeed make $n_x^2\sim t^{2d}$. Our result \eqref{improve sclaing} shows that this does not happen in the translation-invariant case. The strategy, roughly, is to track the local energy accumulation directly instead of relying on particle accumulation as an indicator. The implementation of this idea comes with another wrinkle, essentially because the time derivative of the local energy (commutator with the Hamiltonian) produces local  correlation functions of the particle number that also need to be controlled in the proof. 

 More precisely, the proof relies on the so-called ASTLO (adiabatic space-time localization observables) method. It is inspired by the construction of propagation observables by Sigal and Soffer in the context of few-body scattering theory \cite{SigSof}, later refined in \cites{skibsted, Herbst1995-ac}. It was also extended to a model of quantum electrodynamics by Bony-Faupin-Sigal \cite{Bony2012-pz} and streamlined by Arbunich-Pusateri-Sigal-Soffer \cite{APSS}.
 The key step for this kind of approach is to obtain differential inequalities for propagation observables using  commutator expansions inspired by semiclassical ideas with the time and distance as the effective large parameter. 
 A many-body version was developed in \cites{MR4419446,MR4470244} and then refined in \cite{LRSZ,SiZh, lemm2023microscopic,Lemm2025-ov} particularly to long-ranged Bose-Hubbard Hamiltonians. In this context, the tracked observables were smoothed-out number operators that track particles moving across the lattice. 
 The form of the ASTLO, together with the CCR relations, allows to recover an important recursive structure for these operators to control their time evolution. 
 
 The main novelty of our proof is to construct and control new kind of ASTLO focused on capturing the evolution of local correlations, instead of particle number. 
The key idea is that local on-site energy can be seen as an instance of \textit{local particle-particle correlations}, and we can control the latter because it is dynamically approximately closed under commutators which allows to develop a similar recursive argument. While prior ASTLOs were \textit{quadratic} in bosonic creation and annihilation operators, the new ASTLO we consider here to control local particle-particle correlations is instead \textit{quartic} in them.

Finally, we emphasize that the previous works \cites{KVS,lemm2023microscopic,LRSZ,Lemm2025-ov} proved general moment bounds
 \begin{align}
     \qmexp{\psi(t),n_x^p\psi(t)}\lesssim  t^{pd}.
 \end{align}
It would be extremely interesting to extend an improvement of the form \eqref{improve sclaing} to higher moments, as this  
 would improve the time scaling of the Lieb-Robinson velocity $\sim t^{d-1}$ found in \cite{KVS} and would thus shed light on the open question of when information acceleration can occur in bosonic systems. For more details, refer to the outlook following the statement of the main result.

\subsection{Organization of the paper.}
The paper is structured as follows.
\begin{itemize}
    \item In \secref{sec notation}, we describe our setup, specifically the type of Hamiltonian and Hilbert space we work with, together with some important notation that we will use across the paper. Afterwards, we specify the assumptions on the Hamiltonian and the initial state, and we state the main result, the bound on correlation propagation and its corollary, the ballistic accumulation of bosonic on-site energies.
    \item In \secref{sec ASTLO}, we define the new ingredient of our proof, the quartic ASTLOs that track local correlations of the particle number. We establish the geometric property that allows us to relate them to the number operator itself.
    \item In \secref{sec proof}, we prove the main result by implementing a bootstrapping scheme for the new correlation ASTLOs.
\end{itemize}

\bigskip
\section{Setup and main result}\label{sec notation}

\subsection{Setup and notation}
Let $\Lambda\equiv\Lambda_L\subset \mathbb Z^d$ be a discrete torus of side length $L$, equipped with the Euclidean metric. 
Consider the following bosonic Hamiltonian,
\begin{align}\label{Hamiltonian}
    H\equiv H_{\Lambda}=\sum_{\substack{x,y\in\Lambda}}J_{x,y} b_x^\dagger b_y+V,
    \end{align}
    acting on the bosonic Fock space $\cF_{\Lambda}$ over the one particle Hilbert space $\ell^2(\Lambda)$.
 The matrix    $J=(J_{x,y})_{x,y\in\Lambda}$ is real and symmetric, and it is commonly referred to as the hopping matrix. Here $b_x^\dagger$ and $b_x$ are the bosonic creation and annihilation operators; therefore, they satisfy the canonical commutation relations, or CCR, 
 \begin{align*}
    [b_x,b_y]=[b_x^\dagger,b_y^\dagger]=0 \qquad \text{and} \qquad [b_x,b_y^\dagger]=\delta_{x,y}.  
 \end{align*}
The local number operator is defined as $n_x:=b_x^\dagger b_x$ for every $x\in\Lambda$. We denote by $n_0$ the local number operator supported at the origin of the lattice. 
 The potential is an arbitrary real function of the local number operators, $V=\Phi(\Set{n_x}_{x\in\Lambda})$. Notice that \eqref{bh hamilt} is a special case of \eqref{Hamiltonian}. Given a certain region $X\subseteq\Lambda$,
\begin{align}
    N_X:=\sum_{x\in X}n_x\notag
\end{align}
defines the number operator restricted to such a region $X$. For ease of notation we write $N$ when $X=\Lambda$. It has been shown in  \cite{MR4470244}*{App. A} that Hamiltonians as in \eqref{Hamiltonian} are self-adjoint on 
\( \mathcal D(H)=\{(\psi_N)_{N\geq 0}\in \cF_{\Lambda}\,:\, \sum_{N\geq 0}\|H \psi_N\|^2<\infty \}  .\)

\subsection{Assumptions on the Hamiltonian}
The following are the assumptions we impose on the Hamiltonian.\\

\textit{Assumption 1: translation invariant Hamiltonian.}
    We assume the Hamiltonian to be translation invariant, i.e., for all $x\in\Lambda$,
\begin{align}
    \cT_x^*H\cT_x=H\notag
\end{align}
where $\cT_x$ denotes the unitary translation operator by $x$, $(\cT_x\psi)(y) = \psi(y+ x \mod L)$. From now on, for ease of notation, we will write $x+y$ instead of $ x+y\mod L$ when it is not ambiguous. Translation invariance for the Hamiltonian implies $J_{x,y}=J_{ x+z,y+z}$ for any $x,y,z\in\Lambda$.
\\

\textit{Assumption 2: power-law decay of hopping matrix.}
Additionally, we require polynomial decay of the hopping matrix, specifically, let $\alpha>3d+1$ and set
\begin{align}\label{CJdef}
    C_{J}:=&	\sup_{x,y\in\Lambda}\abs{J_{x,y}}(1+\abs{x-y})^{\alpha},
\end{align}
where $C_J$ is independent of $\Lambda$. 
It follows from \eqref{CJdef} that
\begin{align}\label{ass CJ}
    \abs{J_{x,y}}\le C_J (1+\abs{x-y})^{-\alpha},\qquad \forall x,\,y\in\Lambda.
\end{align}
Thus, the moments of the hopping matrix
\begin{align}\label{kappa b}
    \kappa^{(\beta+1)}:=\sup_{x\in\Lambda}\sum_{y\in\Lambda}\abs{J_{x,y}}\abs{x-y}^{\beta+1},
\end{align}
are bounded by a constant, uniformly in $\Lambda$, for all $\beta< \alpha-d-1$.
 The first moment of the hopping matrix is a key quantity, and we denote it as follows
\begin{align}\label{kappa}
    \kappa\equiv \kappa^{(1)}:=\sup_{x\in\Lambda}\sum_{y\in\Lambda}\abs{J_{x,y}}\abs{x-y}.
\end{align}

\subsection{Assumptions on the initial states}
As we anticipated in the introduction, when working with bosonic systems, we can only expect to obtain bounds on expectation values. This is because we cannot hope to control the full Hilbert space, as initial states with many bosons accumulating on a small number of lattice sites could lead to accelerated dynamics. Our result holds for initial states that satisfy the following assumptions. \\

 \textit{Assumption 3: translation-invariant initial state. }We restrict our attention to translation invariant initial states $\psi_0\in\mathcal D(H)$. Then the assumption of translation invariance of Hamiltonian implies  ${\psi_t=e^{-itH}\psi_0}$ is translation invariant as well. We denote by $\langle\cdot\rangle_t=\langle\psi_t,\cdot\psi_t\rangle$ the corresponding expectation value. \\
 
    \textit{Assumption 4: controlled-density initial state.} We require initial states to satisfy an assumption of controlled density, in other words, there exists $\lambda>0$, such that for $p=1,2$,
    \begin{align}\label{CD2}
    \qmexp{N_{B_r}^p}_0\le (\lam r^d)^p, \qquad  r\ge0.    \end{align}
Here $B_r:=\Set{y\in\Lambda\text{ s.t. }\abs{y}\le r}$, denotes the discrete ball of radius $r$ with center at the origin, and $B_0=\Set{0}$.

We denote by $\mathcal D_T$ the set of states  $\psi\in\mathcal D(H)\cap \mathcal D(N) $ that are translation invariant, and by $\mathcal D^\lambda_T$ the set of states in $\mathcal D_T$ that additionally satisfy \eqref{CD2} for a given $\lambda>0$. 
We remark that
\begin{itemize}
    \item These assumptions are physically relevant, as Mott states are included in this class of initial states.
    \item Translation invariance is essential for our proof, as it is a fundamental ingredient for recovering the recursive structure in our inequalities,  which is a key step.
\end{itemize}

\subsection{Main result}

Our goal is to prove the following bound on local particle-particle correlations. 
\begin{thm}[Bound on correlations]\label{prop corr tr inv}
    Assume \eqref{ass CJ} holds for some $\alpha>3d+1$ and define $\beta:=\lfloor\alpha-d-1\rfloor$. Consider $\lambda>0$ and $\psi_0\in\mathcal D^\lambda_T $. Then, for any $v>2\kappa$ there exists
    \begin{align}
    &C=(v,d,\lambda, \alpha,C_J),  \notag
    \end{align}
     such that, for all $R>r\ge 0$ satisfying $R-r>\max\Set{v,1}$, it holds
\begin{align}\label{stat corr tr inv}
    \sup_{vt\le R-r}\qmexp{n_0 N_{B_r}}_t&\le  \del{1+\frac{C}{R-r}+\frac{C}{(R-r)^{\beta-2d}}}R^d\lam^2.
\end{align}

\end{thm}

\begin{rmk}
  Notice that for $r=0$ and invoking translation invariance, \eqref{stat corr tr inv} implies 
  \begin{align}\label{bound on rep en}
       \qmexp{n_x^2}_t\lesssim R^d \qquad \text{for every }x\in\Lambda.
  \end{align}
 Such an inequality can be interpreted as a bound on the local repulsive energy when considering a Hamiltonian of the kind \eqref{bh hamilt}.
\end{rmk}
\begin{rmk}
    The assumption on the decay of the hopping elements can be improved from $\alpha>3d+1$ to $\alpha>d+1$ by employing a multi-scale analysis as in \cite{lemm2023microscopic} and substituting \thmref{theo contr rem d+1} with \cite{Lemm2025-ov}*{Theorem 2.1}. Notice that in the latter, a lower bound on the the density of the initial states is required. Such an assumption can be lifted by post-processing.
\end{rmk}
\textbf{Outlook.} As mentioned in the introduction, for $R\sim t$, our bound \eqref{bound on rep en} leads to a  $t^d$ bound on $\qmexp{n_x^2}_t$, which compares favorably to the usual particle propagation bounds, which would give a bound of the form $\qmexp{n_x^2}_t\lesssim t^{2d}$ (refer to \cites{lemm2023microscopic,Lemm2025-ov, KVS,LRSZ,YL}). If an improved $t$-scaling persists for higher moments of $n_x$, then this would have powerful implications for the Lieb-Robinson velocity of Bose-Hubbard Hamiltonians for controlled-density initial states in translation-invariant  settings.
    
 Indeed, a crucial step in \cite{KVS} consists in controlling the spectral projector $\qmexp{\Pi_{n_x\ge q}}_t$ onto the states with occupation number greater than a certain $q$ to bound the error between the original and a suitably truncated dynamics. Now, suppose that one could show, e.g., that also $p$th moments accumulate at most ballistically, i.e., $\qmexp{n_x^p}_t\lesssim t^d$ for arbitrary high moment $p$ under appropriate assumptions.
 Then, by Markov's inequality, one could obtain
\begin{align*}
    \qmexp{\Pi_{n_x\ge q}}_t\le \frac{\qmexp{n_x^p}_t}{q^p}\le C_p\del{\frac{t^{d/p}}{q}}^p=C_{d/\epsilon}\del{\frac{t^{\epsilon}}{q}}^{d/\epsilon},
\end{align*}
with $\epsilon=d/p$ arbitrary small. This would imply that the local boson number can be effectively truncated at level $q\sim t^\epsilon$. 
As the Lieb-Robinson velocity would be proportional to the local occupation number, one would obtain $v\sim q\sim t^\epsilon$ for an arbitrarily small $\epsilon$, which would be almost ballistic  and much less than the $v\sim t^{d-1}$ bound shown in \cite{KVS}. We emphasize that it is presently not clear under what circumstances one would physically expected a bound such as $\qmexp{n_x^p}_t\lesssim t^d$ for all moments.

Finally, another interesting open question is whether the bound on the second moment proved here extends to non-translation-invariant settings.

\bigskip
\section{Correlation ASTLO}\label{sec ASTLO} 
Consider $v>2\kappa$, where $\kappa$ is defined as in \eqref{kappa}, and define the following two quantities
\begin{align}
    \tilde v&:=\frac{\kappa+v}{2}\in(\kappa,v),\label{tilde v}\\
    \epsilon&:=v-\tilde v>0.  \label{epsilon}
\end{align}
We introduce the class of cutoff functions, essential to build the ASTLO,
\begin{align}\label{function class}
  \mathcal{C}\equiv \cC_\epsilon:=\left\{
    f\in C^{\infty} (\mathbb{R})\left\vert \begin{aligned} &f\geq0,\: f\equiv 0 \text{ on } (-\infty,\epsilon/2],f\equiv 1 \text{ on } [\epsilon,\infty) \\ &f'\geq 0,\:\sqrt{f'}\in C_c^{\infty} (\mathbb{R}),\:\mathrm{supp} f'\subset(\epsilon/2,\epsilon)
    \end{aligned}\right.\right\}.  
\end{align}
Let $R>r\ge 0$ and write $s:=(R-r)/v$. Then, for every  $t\ge 0$ and every function $\chi\in \cC$ we define a new function $\chi_{ts}:\mathbb R^d \to [0,1]$ by
\begin{align}
    \chi_{ts}(x)=\chi\left(\frac{R-\tilde vt-|x|}{s}\right),\notag\qquad x\in\mathbb R^d.
\end{align}
 Notice that
\begin{align}\label{derivative chi}
    \chi_{ts}'(0)=0,
\end{align}
where, for ease of notation, $\chi'_{ts}$ has to be understood as $(\chi')_{ts}$.
Equality \eqref{derivative chi} holds true since
\begin{align}
    \chi_{ts}'(0)=\chi'\del{\frac{R-\tilde vt}{s}},\notag
\end{align}
and it can be checked that $(R-\tilde vt)/s>\epsilon$ and  $\mathrm{supp}\,\chi'\subset(\epsilon/2,\epsilon)$.
We define a new ASTLO (adiabatic space-time localization observable) that measures particle-particle correlations,
\begin{align}
  \mathbb A_{ts}:=\sum_{x\in\Lambda}\chi_{ts}(x) n_0 n_x.\notag  
\end{align}
It is essential to design the correct ASTLO. For example:
\begin{itemize}
    \item         $\sum_x \chi_{ts}(x) n_x^2$ does not allow to close the recursive structure.
    \item  $\sum_{x,y} \chi_{ts}(x)\chi_{ts}(y) n_x n_y$ would lead to the usual time scaling as $t^{2d}$ since already at initial time it is of size $r^{2d}$.
\end{itemize}
These two aspects create obstacles to controlling higher moments.

\begin{prop}[Geometric properties of the ASTLO]\label{prop geom prop}
For any $v>\kappa$, $\chi\in\cC$,  $R>r>0$, $\lambda>0$, and $\psi_0\in\mathcal D^\lambda_T $, the following holds
\begin{align}\label{ftc in prop}
    \qmexp{n_0 N_{B_{r}}}_t\le R^d\lam^2+\int_0^t \frac{d}{d\tau} \qmexp{\A_{\tau s}}_\tau d\tau,
\end{align}
for every $t\le s$, where $s:=(R-r)/v$. 
\end{prop}
\begin{proof}
    Notice that, thanks to the properties of $\cC$, it holds that
    \begin{align}\label{eq:comparechiind}
        \mathbbm 1_{r\geq \epsilon}\leq \chi(r)\leq \mathbbm 1_{r\geq \frac{\epsilon}{2}}.
    \end{align}
  Recalling the definitions for $\tilde v$ \eqref{tilde v} and $\epsilon$ \eqref{epsilon},  for every $t\le s$,
    \begin{align}\label{sets inclusion bound below}
        \Set{\abs{x} \le R -\tilde vt-\epsilon s}\supset \Set{\abs{x} \le R -\tilde vs-vs+\tilde vs}=\Set{\abs{x} \le R -vs}=\Set{\abs{x}\le r}.
    \end{align}
      Lines \eqref{eq:comparechiind} and \eqref{sets inclusion bound below} imply that the time-$t$ ASTLO $\mathbb A_{ts}$ satisfies, for every $t\le s$, the following lower bound
    \begin{align}
        \qmexp{\A_{ts}}_t\geq
    \sum_{\abs{x} \le R -\tilde vt-\epsilon s} \qmexp{n_0 n_x}_t \ge \sum_{\abs{x} \,\le\, r} \qmexp{n_0 n_x}_t= \qmexp{n_0 N_{B_{r}}}_t.\notag
    \end{align}
    On the other hand, by \eqref{eq:comparechiind}, the time-zero ASTLO $\mathbb A_{0s}$ satisfies 
     \begin{align}
        \qmexp{\A_{0s}}_0\leq
    \sum_{|x|\leq R-\epsilon s/{2}} \qmexp{n_0 n_x}_0\le \sum_{|x|\leq R} \qmexp{n_0 n_x}_0=\qmexp{n_0 N_{B_R}}_0,\notag
    \end{align}
    where in the last inequality, we used the fact that the local number operator is positive. The Cauchy-Schwarz inequality yields
    \begin{align}\label{ub astlo before dens contr}
        \qmexp{\A_{0s}}_0\leq \sqrt{\qmexp{n_0^2}_0\qmexp{N^2_{B_R}}_0}.
    \end{align}
    Given the assumption of controlled density \eqref{CD2}, inequality \eqref{ub astlo before dens contr} implies
    \begin{align}
        \qmexp{\A_{0s}}_0\leq \lam^2 R^d.\notag
    \end{align}
      That is, the time-zero ASTLO is bounded by a constant times $R^d$. We now combine these facts. By the fundamental theorem of calculus, we have
    \begin{equation}\label{eq:ftc}
        \qmexp{n_0 N_{B_{r}}}_t
      \leq  \qmexp{\A_{ts}}_t
      \leq  \qmexp{\A_{ts}}_t+\lam^2 R^d-\qmexp{\A_{0s}}_0
    =\lam^2 R^d+\int_0^t \frac{d}{d\tau} \qmexp{\A_{\tau s}}_\tau d\tau.
    \end{equation}  
\end{proof}

The main work to derive \thmref{prop corr tr inv} is to control the derivative appearing on the r.h.s. of \eqref{ftc in prop}, namely
\[
\frac{d}{d t} \qmexp{\A_{t s}}_t=\qmexp{\partial_t\A_{t s}+[iH,\A_{ts}]}_t .
\]
This is the content of the next and final section.

\bigskip

\section{Proof of the main result}\label{sec proof}
In this section, we prove \thmref{prop corr tr inv}. We start by showing the following bound, displaying the necessary recursive structure. From now on, the constants can change from line to line, still remaining independent of $r,R$, the lattice size, and the total number of particles.

 \begin{prop}[Differential inequality]\label{prop diff ineq}
    Assume \eqref{ass CJ} holds for some $\alpha>3d+1$, and denote $\beta=\lfloor \alpha-d-1\rfloor$.
Consider $\chi\in\mathcal C$. Then there exist $C=C(\chi,\alpha,d,C_J)>0$ and $\tilde \chi\in\mathcal C$ such that for every $v>2\kappa$, $\psi_0\in\mathcal{D_T}$, and ${R>r>0}$ with $R-r>\max \Set{v,1}$, the following holds
\begin{align}\label{contr der w rem}
\qmexp{\partial_t\A_{t s}+[iH,\A_{ts}]}_t
\leq \frac{2\kappa-\tilde v}{s}\qmexp{\A'_{t s}}_t+C\frac{1}{s^2}\qmexp{\tilde \A'_{t s}}_t+\frac{C}{s^{\beta+1}}R^d\qmexp{n_0^2}_t.
\end{align}
Recall $\tilde v$ is as in \eqref{tilde v} and we define
\begin{align}
    \A'_{ts}&:=\sum_{x\in\Lambda}\chi'_{ts}(x) n_0 n_x,\notag\\
   \tilde \A'_{ts}&:=\sum_{x\in\Lambda}\tilde\chi'_{ts}(x) n_0 n_x.\notag
\end{align}

\end{prop}
\begin{rmk}
    The assumption of translation invariance is necessary to prove \propref{prop diff ineq}, as it allows to recover the ASTLO structure. 
\end{rmk}

Combining \propref{prop geom prop} with \propref{prop diff ineq}, and subsequently repeating such a procedure, we obtain the following proposition.
\begin{prop}[Bootstrapping] \label{bootstrapping}
    In the same setting as of \propref{prop diff ineq}, there exists a constant $C=C(v,\alpha,d,C_J)$ such that, for every $\lambda>0$ and $\psi_0\in\mathcal{D_T^\lambda}$, the following holds
 \begin{align}\label{contr int astlo}
    \qmexp{n_0 N_{B_r}}_t\le R^d\lambda^2\del{1+Cs^{-1}}+\frac{Ct}{s^{\beta+1}}R^d\sup_{vt\le R-r}\qmexp{n_0^2}_t.
 \end{align}    
 
 \end{prop}

 To control the remainder term appearing in \eqref{contr int astlo}, we employ the following theorem, whose proof can be found in \cite{Lemm2025-ov}.
 \begin{thm}[\cite{lemm2023microscopic}*{Theorem 2.1}]\label{theo contr rem d+1}
Let $p\ge1$ be an integer and assume that \eqref{ass CJ} holds with 
				\begin{align}
					{\alpha>\max\{3dp/2+1,2d+1\}},\notag
				\end{align}
                and define the following quantity
                \begin{align*}
n:=\left\{ \begin{aligned} &\lfloor\alpha-d-1\rfloor,\quad &\textnormal{for }p=1, \\ &\lfloor\alpha-{\tfrac{dp}{2}}-1\rfloor,\quad          
                            &\textnormal{for } p\ge2.
				\end{aligned}\right.
			\end{align*}
				Then, for any ${v} > \kappa$
				{and $\delta_0>0$, there exists a positive constant ${C=C(\alpha,d, C_J, v,\delta_0,p)}$ such that for all $\lambda,\,r_2,\,r_1>0$ with $r_2-r_1>\max(\delta_0r_1,1)$  and initial states ${\psi_0\in\mathcal D(H_{\Lambda})\cap \cD(N^{p/2})}$ satisfying \eqref{CD2},}
				\begin{align*}
					\sup_{0\le{t}<(r_2-r_1)/v}\qmexp{N_{B_{r_1}}^p}_t\le {\left(1+ C(r_2-r_1)^{-1}\right)} {{\qmexp{N_{B_{r_2}}^p}_0}} + C(r_2-r_1)^{-n+dp}\lambda^p.
				\end{align*}
                
  \end{thm}
  Having prepared all necessary ingredients, we can finally prove \thmref{prop corr tr inv}.
  \begin{proof}[Proof of \thmref{prop corr tr inv}]
    We remark that the Hamiltonian and the initial state satisfy the required assumptions. Then, to bound the remainder term appearing in \eqref{contr int astlo}, we can apply \thmref{theo contr rem d+1} for $p=2$, $\delta_0=1$, $r_1=1$ and  ${r_2=R-r+1}$. Thus, for ${vt\le R-r}$, it holds 
    \begin{align}\label{bound rem w theo}
       \qmexp{n_0^2}_t &\le \qmexp{n_{B_1}^2}_t\notag\\
       &\le{\left(1+ C(R-r)^{-1}\right)} {{\qmexp{N_{B_{R-r+1}}^2}_0}} + C(R-r)^{-\beta+2d}\lambda^2\notag\\
        &\le C \lam^2\del{R-r}^{2d}.
    \end{align}
    In the first inequality, we applied \thmref{theo contr rem d+1}, and in the second, we used \eqref{CD2}.
   We apply \eqref{bound rem w theo} to control the remainder term appearing in \eqref{contr int astlo} 
    \begin{align}
        \qmexp{n_0 N_{B_r}}_t&\le R^d\lambda^2\del{1+\frac{C}{s}}+\frac{C t}{s^{\beta+1}}R^d\sup_{vt\le R-r}\qmexp{n_0^2}_t\notag\\
        &\le R^d\lam^2\del{1+\frac{C}{s}}+\frac{C\lam^2t}{s^{\beta+1}}R^d\del{R-r}^{2d}.\notag
    \end{align}
    Recalling the fact that  $s=(R-r)/v$ and $t\le s$, we obtain \eqref{stat corr tr inv}.
    \end{proof}

\begin{proof}[Proof of \propref{prop diff ineq}]
The chain rule implies that $
\partial_t\A_{t s}=-\frac{\tilde v}{s}\A'_{t s}$, so it suffices to prove
\begin{equation}
    \qmexp{[iH,\A_{ts}]}_t\leq \frac{2\kappa}{s}\qmexp{\A'_{t s}}_t+C\frac{1}{s^2}\qmexp{\tilde \A'_{t s}}_t+Cs^{-{\beta+1}}R^d\qmexp{n_0^2}_t.\notag
\end{equation}
We calculate, using the CCR,

\begin{align}\label{I + II}
     & [iH,\A_{ts}]\notag\\
 =& -i \sum_{x,y\in\Lambda}J_{x,y}[b_x^\dagger b_y,\A_{ts}]\notag\\
 =&-i \sum_{x,y\in\Lambda}J_{x,y}\sum_{z\in\Lambda}
\chi_{ts}(z) [b_x^\dagger b_y,n_zn_0]\notag\\
 =&-i \sum_{x,y\in\Lambda}J_{x,y}\sum_{z\in\Lambda}
\chi_{ts}(z) 
\left([b_x^\dagger b_y,n_z]n_0+n_z[b_x^\dagger b_y,n_0]
\right)\notag\\
 =&-i \sum_{x,y\in\Lambda}J_{x,y}\sum_{z\in\Lambda}
\chi_{ts}(z) 
\left((\delta_{y,z} b_x^\dagger b_z-\delta_{z,x}b_z^\dagger b_y)n_0+n_z (\delta_{y,0} b_x^\dagger b_0-\delta_{0,x}b_0^\dagger b_y)
\right)\notag\\
 =&-i \Big(\underbrace{\sum_{z\in\Lambda}\chi_{ts}(z) \sum_{x\in\Lambda} J_{x,z} (b_x^\dagger b_z-b_z^\dagger b_x)n_0}_{=:\mathrm{(I)}}+\underbrace{\sum_{z\in\Lambda}\chi_{ts}(z)\sum_{x\in\Lambda\setminus\Set{0}}J_{x,0}n_z ( b_x^\dagger b_0-b_0^\dagger b_x)}_{=:\mathrm{(II)}}
\Big),
\end{align}
where in the last step we relabeled the sums suitably and the fact that $J_{x,y}=J_{y,x}$ for every $x,y\in\Lambda$. Terms $\mathrm{(I)}$ and $\mathrm{(II)}$ have to be treated rather differently.\\

\underline{Term $\mathrm{(I)}$.}
We can relabel term $\mathrm{(I)}$ to produce differences of $\chi_{ts}$, which is essential to run the ASTLO method. Namely,
\[
\begin{aligned}
\mathrm{(I)}
=\sum_{y\in\Lambda} \chi_{ts}(y) \sum_{x\in\Lambda}  J_{x,y}(b_x^\dagger b_y-b_y^\dagger b_x)n_0
=\sum_{x,y\in\Lambda} J_{x,y}(\chi_{ts}(x)-\chi_{ts}(y))b_x^\dagger b_yn_0
\end{aligned}
\]
First, to prepare for applying Cauchy-Schwarz later, we move the $n_0$ into the center up to a commutator, which will be sub-leading for our purposes:
\begin{equation}\label{eq:centern}
b_x^\dagger b_yn_0=
b_x^\dagger n_0 b_y+\delta_{y,0} b_x^\dagger b_0.
\end{equation}
This gives
\begin{equation}\label{eq:Irewrite}
    \begin{aligned}
\mathrm{(I)}
=\sum_{x,y\in\Lambda} J_{x,y}(\chi_{ts}(x)-\chi_{ts}(y))(b_x^\dagger n_0b_y+\delta_{y,0} b_x^\dagger b_0).
\end{aligned}
\end{equation}
A key ingredient of the ASTLO approach is the symmetrized expansion (\cite{MR4470244}*{Lemma 2.2}) with (a refinement of) the a priori localization trick from \cite{lemm2023microscopic}* {Eq.(4.35)--(4.36)}.
More precisely, notice that $\chi_{ts}(x)-\chi_{ts}(y)\neq 0$ implies that $x$ or $y$ lie in $\mathrm{supp}\, \chi_{ts}\subset B_{R}$.
Then, for every $\beta$, there exist $j_2,\ldots,j_\beta\in\mathcal C$ such that we have the symmetrized expansion
\begin{equation}\label{eq:symmexp}
\begin{aligned}
&|\chi_{ts}(x)-\chi_{ts}(y)|\\
\leq& |\chi_{ts}(x)-\chi_{ts}(y)| \del{\mathbbm 1_{|x|\leq R }+\mathbbm 1_{|y|\leq R}}\\
\leq& 
\frac{|x-y|}{s} u_{ts}(x)u_{ts}(y)+\sum_{k=2}^\beta\frac{|x-y|^k}{s^k} C_{\chi,k}
u_{k,ts}(x)u_{k,ts}(y)\\
&+C_{\chi,\beta+1}\frac{|x-y|^{\beta+1}}{s^{\beta+1}}  \del{\mathbbm 1_{|x|\leq R }+\mathbbm 1_{|y|\leq R}},
\end{aligned}
\end{equation}
where the sum should be dropped for $\beta=1$, $\mathbbm 1_X$ is the characteristic function of $X\subset\Lambda$, and
\[
u_{ts}=(\sqrt{\chi'})_{ts},\qquad u_{k,ts}\equiv (u_k)_{ts}=(\sqrt{j_k'})_{ts},\quad (k=2,\ldots,\beta).
\]
Additionally, for ease of notation, we write $(\chi')_{ts}$ and $(j'_k)_{ts}$  as $\chi'_{ts}$ and $j'_{k,ts}$ respectively.
We test \eqref{eq:Irewrite} on the vector $\psi_t$ and apply the above expansion to obtain
\begin{equation}\label{eq:IpostCS}
    \begin{aligned}
    \qmexp{\mathrm{(I)}}_t
\leq& \frac{1}{s}\sum_{x,y\in\Lambda} \abs{J_{x,y}}\abs{x-y}u_{ts}(x)u_{ts}(y)\left(|\qmexp{b_x^\dagger n_0b_y}_t|+\delta_{y,0}|\qmexp{b_x^\dagger b_y}_t|\right)\\
&+C\sum_{k=2}^\beta s^{-k} \sum_{x,y\in\Lambda} \abs{J_{x,y}}\abs{x-y}^ku_{k,ts}(x)u_{k,ts}(y)\left(|\qmexp{b_x^\dagger n_0b_y}_t|+\delta_{y,0}|\qmexp{b_x^\dagger b_y}_t|\right)\\
&+C s^{-\beta-1}\mathrm{Rem'}\\
=& \frac{1}{s}\sum_{x,y\in\Lambda} \abs{J_{x,y}}\abs{x-y}u_{ts}(x)u_{ts}(y)|\qmexp{b_x^\dagger n_0b_y}_t|\\
&+C\sum_{k=2}^\beta s^{-k} \sum_{x,y\in\Lambda}\abs{J_{x,y}}\abs{x-y}^k u_{k,ts}(x)u_{k,ts}(y)|\qmexp{b_x^\dagger n_0b_y}_t|\\
&+C s^{-\beta-1}\mathrm{Rem'},
\end{aligned}
\end{equation}
with
\begin{align}\label{long rem}
  \mathrm{Rem'}:=  \sum_{x,y\in \Lambda} \del{\mathbbm 1_{|x|\leq R }+\mathbbm 1_{|y|\leq R}}\;\abs{J_{x,y}}\abs{x-y}^{\beta+1}\left(|\qmexp{b_x^\dagger n_0b_y}_t|+\delta_{y,0}|\qmexp{b_x^\dagger b_y}_t|\right).
\end{align}
Notice that in the second step of \eqref{eq:IpostCS}, we used that $\chi'_{ts}(0)=j_{k,ts}'(0)=0$ for $\chi,j_k\in\mathcal C$.
Recall
\begin{align}\label{bnb le nn}
    b_x^\dagger n_0b_x\leq n_xn_0.
\end{align} 
By applying subsequently  Cauchy-Schwarz, \eqref{bnb le nn}, $\chi'_{ts}(0)=j_{k,ts}'(0)=0$, and Cauchy-Schwarz again, we can control the first summand in \eqref{eq:IpostCS} as follows
\begin{align}\label{first I ok}
  \frac{1}{s}\sum_{x,y\in\Lambda}& \abs{J_{x,y}}\abs{x-y}u_{ts}(x)u_{ts}(y)|\qmexp{b_x^\dagger n_0b_y}_t|\notag\\
  &\le  \frac{1}{s}\sum_{x,y\in\Lambda} \abs{J_{x,y}}\abs{x-y}u_{ts}(x)u_{ts}(y)\sqrt{\qmexp{b_x^\dagger n_0b_x}_t\qmexp{b_y^\dagger n_0b_y}_t}\notag\\
  &\le\frac{1}{s}\sum_{x,y\in\Lambda} \del{\abs{J_{x,y}}\abs{x-y}\chi'_{ts}(x)\qmexp{n_xn_0}_t}^{1/2}\del{\abs{J_{x,y}}\abs{x-y}\chi'_{ts}(y)\qmexp{n_yn_0}_t}^{1/2}\notag\\
  &\le \frac{\kappa}{s}\sum_{x} \chi'_{ts}(x)\qmexp{n_xn_0}_t.
\end{align}
Following the same strategy  we can control the second summand in \eqref{eq:IpostCS}
\begin{align}\label{second I almost ok}
    \sum_{k=2}^\beta s^{-k} &\sum_{x,y\in\Lambda}\abs{J_{x,y}}\abs{x-y}^k u_{k,ts}(x)u_{k,ts}(y)|\qmexp{b_x^\dagger n_0b_y}_t|\notag \\
    &\le \sum_{k=2}^\beta \frac{\kappa^{(k)}}{s^{k}} \sum_{x} j'_{k,ts}(x)\qmexp{n_xn_0}_t.
\end{align}
Recall in \eqref{kappa b} we defined $\kappa^{(k)}:= \sup_{x\in\Lambda} \sum_{y\in\Lambda}\abs{J_{x,y}}\abs{x-y}^k$. 
Let us denote
\[
\tilde{\A}'_{k,ts}:=\sum_{x\in\Lambda_L}j'_{k,ts}(x) n_0 n_x,
\]
and recall the following property  of $\tilde \cE$,
\begin{align}\label{prop f in C}
    \forall \;f_1,\,f_2\in\cC\quad\exists f_3\in\cC \text{ and }\tilde C>0 \text{ such that }f_1+f_2\le \tilde C f_3.
\end{align}
Then, for $s>1$, it follows
\begin{align}\label{secon I ok}
   \sum_{k=2}^\beta \frac{\kappa^{(k)}}{s^{k}} \sum_{x} {j_{k,ts}'}(x)\qmexp{n_xn_0}_t = \sum_{k=2}^\beta\frac{\kappa^{(k)}}{s^{k}} \qmexp{\tilde{\A}'_{k,ts}}_t\le Cs^{-2} \qmexp{\tilde{\A}'_{ts}}_t,
\end{align}
for a suitably defined $\tilde\chi\in\mathcal C$.
Combining \eqref{eq:IpostCS} together with \eqref{first I ok}, \eqref{second I almost ok}, and \eqref{secon I ok} yields
\begin{equation}\label{I astlo}
    \begin{aligned}
  \qmexp{\mathrm{(I)}}_t
\leq&\frac{\kappa}{s}\qmexp{\mathbb A'_{ts}}_t+Cs^{-2} \qmexp{\tilde{\A}'_{ts}}_t+C s^{-\beta-1}\mathrm{Rem'}.
\end{aligned}
\end{equation}
 Notice that for $\beta=1$, since we would drop the sum over $k$, the second summand of the r.h.s. of \eqref{I astlo} has to be dropped.

Let us now analyze the remainder term $\mathrm{Rem'}$ defined as in \eqref{long rem}. We apply Cauchy-Schwarz once again to obtain
\begin{align}\label{eq rem 1}
    \mathrm{Rem'}:=&  \sum_{x,y\in \Lambda} \del{\mathbbm 1_{|x|\leq R }+\mathbbm 1_{|y|\leq R}}\;\abs{J_{x,y}}\abs{x-y}^{\beta+1}\left(|\qmexp{b_x^\dagger n_0b_y}_t|+\delta_{y,0}|\qmexp{b_x^\dagger b_y}_t|\right)\notag\\
    \le& \sum_{x,y\in \Lambda}\del{ \mathbbm 1_{|x|\leq R }+\mathbbm 1_{|y|\leq R}}\;\abs{J_{x,y}}\abs{x-y}^{\beta+1}\left(\sqrt{\qmexp{n_x n_0}_t\qmexp{n_y n_0}_t}+\delta_{y,0}\sqrt{\qmexp{n_x}_t\qmexp{n_y}_t}\right)\notag\\
    \le& \frac{1}{2}\sum_{x,y\in \Lambda}\del{ \mathbbm 1_{|x|\leq R }+\mathbbm 1_{|y|\leq R}}\;\abs{J_{x,y}}\abs{x-y}^{\beta+1}\left({\qmexp{n_x n_0}_t+\qmexp{n_y n_0}_t}+\delta_{y,0}\del{\qmexp{n_x}_t+\qmexp{n_y}_t}\right).
  \end{align}
From inequality \eqref{eq rem 1} we obtain
\begin{align}\label{A and B}
  \mathrm{Rem'} \le    \underbrace{\sum_{x\in B_R,\,y\in \Lambda}\abs{J_{x,y}}\abs{x-y}^{\beta+1}\del{{\qmexp{n_x n_0}_t+\qmexp{n_y n_0}_t}}}_{\Mr{(A)}}+\frac{1}{2}\underbrace{\sum_{x\in\Lambda}\abs{J_{x,0}}\abs{x}^{\beta+1}\del{\qmexp{n_x}_t+\qmexp{n_0}_t}}_{\Mr{(B)}}.
\end{align}
Thanks to translation invariance, we can control term (B) in \eqref{A and B} as follows, 
\begin{align}
    {\Mr{(B)}}=& \qmexp{n_0}_t\sum_{x\in\Lambda}\abs{J_{x,0}}\abs{x}^{\beta+1}\le \qmexp{n_0}_t\kappa^{(\beta+1)}.\label{eq rem 2}
\end{align}
 Now let us focus on term (A) in \eqref{A and B}. Applying translation invariance and Cauchy-Schwarz leads to
\begin{align}\label{eq rem 3}
    {\Mr{(A)}}&\le \sum_{x\in B_R,\,y\in \Lambda}\abs{J_{x,y}}\abs{x-y}^{\beta+1}\del{\sqrt{\qmexp{n_x^2}_t\qmexp{n_0^2}_t}+\sqrt{\qmexp{n_y^2}_t\qmexp{n_0^2}_t}}\notag\\
    &= 2\qmexp{n_0^2}_t \sum_{x\in B_R,\,y\in \Lambda}\abs{J_{x,y}}\abs{x-y}^{\beta+1}\notag\\
    &\le 2\kappa^{(\beta+1)}\abs{B_R}\qmexp{n_0^2}_t.
\end{align}
  Combining \eqref{eq rem 2} and \eqref{eq rem 3}, we obtain
  \begin{align}
    \mathrm{Rem'} \le& \kappa^{(\beta+1)} \del{\qmexp{n_0}_t+2\abs{B_R}\qmexp{n_0^2}_t}\notag\\
    \le& 3\kappa^{(\beta+1)}\abs{B_R}\qmexp{n_0^2}_t.\notag
  \end{align}
  Applying this bound for the remainder to \eqref{I astlo} yields
  \begin{align}\label{I astlo no rem}
    \qmexp{\mathrm{(I)}}_t
\leq&\frac{\kappa}{s}\qmexp{\mathbb A'_{ts}}_t+Cs^{-2} \qmexp{\tilde{\A}'_{ts}}_t+\frac{C} {s^{\beta+1}}\kappa^{(\beta+1)}\abs{B_R}\qmexp{n_0^2}_t.
  \end{align}
\underline{Term $\mathrm{(II)}$.} We consider the term

\begin{align*}
\qmexp{\mathrm{(II)}}_t=&
\sum_{y\in\Lambda} \chi_{ts}(y)\sum_{x\in\Lambda\setminus\Set{0}}J_{x,0} \qmexp{n_y ( b_x^\dagger b_0-b_0^\dagger b_x)}_t\\
&=\sum_{y\in\Lambda} \chi_{ts}(y)\sum_{a=1}^{L/2}\sum_{j=1}^d\sum_{\sigma\in\{\pm 1\}}J_{a\sigma e_j,0}
\left(\qmexp{n_y  b_{-a\sigma e_j}^\dagger b_0}_t-\qmexp{n_y b_0^\dagger b_{a\sigma e_j}}_t\right)\\
&=\sum_{y\in\Lambda} \chi_{ts}(y)\sum_{a=1}^{L/2}\sum_{j=1}^d\sum_{\sigma\in\{\pm 1\}}J_{a\sigma e_j,0}
\left(\qmexp{n_{y+a\sigma e_j}  b_{0}^\dagger b_{a\sigma e_j}}_t-\qmexp{n_{y} b_{0}^\dagger b_{a\sigma e_j}}_t\right)\\
&=\sum_{y\in\Lambda}\sum_{x\in\Lambda\setminus\Set{0}} J_{x,0} \left(\chi_{ts}(y-x)-\chi_{ts}(y)\right)
\qmexp{n_y  b_{0}^\dagger b_{x}}_t,\\
\end{align*}
where we used translation-invariance in the second-to-last step.
Notice that the assumption of translation invariance, together with the symmetry of the Hamiltonian, implied 
\begin{align}
   J_{-a\sigma e_j,0}=J_{0,a\sigma e_j}=J_{a\sigma e_j,0}.\notag
\end{align}
As in \eqref{eq:centern}, we move the $n$-operator to the center:
\[
n_y  b_{0}^\dagger b_{x}=  b_{0}^\dagger n_yb_{x}+\delta_{y,0} b_0^\dagger b_{x}
\]
which gives
\begin{align}
\qmexp{\mathrm{(II)}}_t
&=\sum_{y\in\Lambda}\sum_{x\in\Lambda\setminus\Set{0}} J_{x,0} \left(\chi_{ts}(y-x)-\chi_{ts}(y)\right)
\left(\qmexp{  b_{0}^\dagger n_y b_{x}}_t+\delta_{y,0} \qmexp{b_0^\dagger b_{x}}_t\right).\notag
\end{align}
Now we are again in a position to perform the symmetrized expansion \eqref{eq:symmexp},
\begin{align}
&|\chi_{ts}(y)-\chi_{ts}(y-x)|\notag\\
\leq& 
\frac{1}{s} u_{ts}(y)u_{ts}(y-x)\abs{x}+\sum_{k=2}^\beta\frac{1}{s^k} C_{\chi,k}
u_{k,ts}(y)u_{k,ts}(y-x)\abs{x}^k\notag\\
&+C_{\chi,\beta+1}\frac{1}{s^{\beta+1}} \abs{x}^{\beta+1}\del{\mathbbm 1_{|y|\leq R }+\mathbbm 1_{|y-x|\leq R}}.\notag
\end{align}
Since $\chi'(0)=j_k'(0)=0$, we obtain
\begin{align}
\qmexp{\mathrm{(II)}}_t
&\leq \frac{1}{s}\sum_{y\in\Lambda}\sum_{x\in\Lambda\setminus\Set{0}} \abs{J_{x,0}}\abs{x}u_{ts}(y)u_{ts}(y-x)
\left|\qmexp{  b_{0}^\dagger n_y b_{x}}_t\right|\notag
\\
&+
C \sum_{k=2}^\beta\frac{1}{s^k}\sum_{y\in\Lambda} \sum_{x\in\Lambda\setminus\Set{0}}\abs{J_{x,0}}\abs{x}^ku_{k,ts}(y)u_{k,ts}(y-x)
\left|\qmexp{  b_{0}^\dagger n_y b_{x}}_t\right|\notag
\\
&+ \frac{C}{s^{\beta+1}}\sum_{y\in \Lambda} \sum_{x\in\Lambda\setminus\Set{0}}\del{\mathbbm 1_{|y|\leq R }+\mathbbm 1_{|y-x|\leq R}}\abs{J_{x,0}}\abs{x}^{\beta+1}
\left(\left|\qmexp{  b_{0}^\dagger n_y b_{x}}_t\right|
+\delta_{y,0} \left|\qmexp{b_0^\dagger b_{x}}_t\right|\right).\notag
\end{align}
Applying Cauchy-Schwarz and \eqref{bnb le nn} we derive the following
\begin{align}\label{bxnby le nn nn}
    \qmexp{  b_{0}^\dagger n_y b_{x}}_t&\le \sqrt{\qmexp{b_{0}^\dagger n_{y} b_{0}}_t \qmexp{b_{x}^\dagger n_{y} b_{x}}_t}\notag\\
    &\le \sqrt{\qmexp{n_0 n_y}_t \qmexp{n_{x} n_{y} }_t}.
\end{align}
Thanks to \eqref{bxnby le nn nn} and $\chi'(0)=j_k'(0)=0$, we have
\begin{align}\label{II w nn}
    \qmexp{\mathrm{(II)}}_t
&\leq \frac{1}{s}\sum_{y\in\Lambda}\sum_{x\in\Lambda\setminus\Set{0}} \abs{J_{x,0}}\abs{x}u_{ts}(y)u_{ts}(y-x)
\sqrt{\qmexp{n_0 n_y}_t \qmexp{n_{x} n_{y} }_t}\notag\\
&+
C \sum_{k=2}^\beta\frac{1}{s^k}\sum_{y\in\Lambda} \sum_{x\in\Lambda\setminus\Set{0}}\abs{J_{x,0}}\abs{x}^ku_{k,ts}(y)u_{k,ts}(y-x)
\sqrt{\qmexp{n_0 n_y}_t \qmexp{n_{x} n_{y} }_t}\notag
\\
&+  \frac{C}{s^{\beta+1}}\mathrm{Rem''},
\end{align}
where we defined
\begin{align}\label{bound on rem''}
    \mathrm{Rem''}:=\sum_{\substack{y\in \Lambda\\x\in\Lambda\setminus\Set{0}}}\del{\mathbbm 1_{|y|\leq R }+\mathbbm 1_{|y-x|\leq R}}\abs{J_{x,0}}\abs{x}^{\beta+1}
    \left(\sqrt{\qmexp{n_0 n_y}_t \qmexp{n_{x} n_{y} }_t}
    +\delta_{y,0} \sqrt{\qmexp{n_0}_t\qmexp{n_{x}}_t}\right).
\end{align}
We apply  Cauchy-Schwarz for numbers to the first row of \eqref{II w nn} to obtain
\begin{align}
    & \frac{1}{s}\sum_{y\in\Lambda}\sum_{x\in\Lambda\setminus\Set{0}} \abs{J_{x,0}}\abs{x}u_{ts}(y)u_{ts}(y-x)
\sqrt{\qmexp{n_0 n_y}_t \qmexp{n_{x} n_{y} }_t}\notag\\
    = & \frac{1}{s}\sum_{y\in\Lambda}\sum_{x\in\Lambda\setminus\Set{0}} \del{\abs{J_{x,0}}\abs{x}\chi'_{ts}(y){\qmexp{n_0 n_y}_t}}^{1/2}\del{ \abs{J_{x,0}}\abs{x}\chi'_{ts}(y-x){\qmexp{n_{x} n_{y} }_t}}^{1/2}\notag\\
    \le & \frac{1}{2s}\sum_{y\in\Lambda}\sum_{x\in\Lambda\setminus\Set{0}} \del{\abs{J_{x,0}}\abs{x}\chi'_{ts}(y){\qmexp{n_0 n_y}_t}+ \abs{J_{x,0}}\abs{x}\chi'_{ts}(y-x){\qmexp{n_{x} n_{y} }_t}}.\label{II a before A}
\end{align}
The first term in \eqref{II a before A} can be reduced to the ASTLOs as follows
\begin{align}\label{II a astlo 1}
    \sum_{y\in\Lambda}\sum_{x\in\Lambda\setminus\Set{0}} \abs{J_{x,0}}\abs{x}\chi'_{ts}(y){\qmexp{n_0 n_y}_t}&=  \sum_{y\in\Lambda} \chi'_{ts}(y){\qmexp{n_0 n_y}_t}\sum_{x\in\Lambda\setminus\Set{0}} \abs{J_{x,0}}\abs{x}\notag\\
    &\le \kappa \qmexp{\bA'_{ts}}_t.
\end{align}
We can treat the second term in \eqref{II a before A} similarly, after using the assumption of translation invariance and rescaling the sum.

\begin{align}\label{II a aslto 2}
    &\sum_{y\in\Lambda}\sum_{x\in\Lambda\setminus\Set{0}} \abs{J_{x,0}}\abs{x}\chi'_{ts}(y-x){\qmexp{n_x n_y}_t}\notag\\
   = &\sum_{y\in\Lambda}\sum_{x\in\Lambda\setminus\Set{0}} \abs{J_{x,0}}\abs{x}\chi'_{ts}(y-x){\qmexp{n_0 n_{y-x}}_t}\notag\\
    = &\sum_{z\in\Lambda}\chi'_{ts}(z){\qmexp{n_0 n_{z}}_t} \sum_{x\in\Lambda\setminus\Set{0}}\abs{J_{x,0}}\abs{x}\notag\\
    \le&\kappa \qmexp{\bA'_{ts}}_t.
\end{align}
Using the same strategy, we can bound the second row of \eqref{II w nn} and obtain 
\begin{align}\label{II b astlo}
   &  \sum_{k=2}^\beta\frac{1}{s^k}\sum_{y\in\Lambda} \sum_{x\in\Lambda\setminus\Set{0}}\abs{J_{x,0}}\abs{x}^ku_{k,ts}(y)u_{k,ts}(y-x)
\sqrt{\qmexp{n_0 n_y}_t \qmexp{n_{x} n_{y} }_t}\notag\\
\le& \sum_{k=2}^\beta\frac{\kappa^{(k)}}{s^k}\qmexp{\tilde\bA'_{k,ts}}_t.
\end{align}
Combining \eqref{II w nn} together with \eqref{II a astlo 1}, \eqref{II a aslto 2}, and \eqref{II b astlo} leads to
\begin{align}
    \qmexp{\mathrm{(II)}}_t
&\leq \frac{\k}{s} \qmexp{\bA'_{ts}}_t+C\sum_{k=2}^\beta\frac{1}{s^k}\qmexp{\tilde\bA'_{k,ts}}_t+\frac{C}{s^{\beta+1}}\mathrm{Rem''}.\notag
\end{align}
Thanks to  property \eqref{prop f in C} of the function class $\cC$, defined in \eqref{function class}, and the fact $s>1$, we obtain
\begin{align}\label{II after boot still rem}
    \qmexp{\mathrm{(II)}}_t
    &\leq \frac{\kappa}{s} \qmexp{\bA'_{ts}}_t+\frac{C}{s^2}\qmexp{\tilde\bA'_{ts}}_t+\frac{C}{s^{\beta+1}}\mathrm{Rem''},
\end{align}
for a certain $\tilde\chi\in\cC$. 
Thanks to translation invariance and Cauchy-Schwarz we can control the remainder in \eqref{II after boot still rem}, defined in \eqref{bound on rem''}, as
\begin{align}
    \mathrm{Rem''} &\le \sum_{y\in \Lambda} \sum_{x\in\Lambda\setminus\Set{0}}\del{\mathbbm 1_{|y|\leq R }+\mathbbm 1_{|y-x|\leq R}}\abs{J_{x,0}}\abs{x}^{\beta+1}\del{\delta_{y,0}\qmexp{n_0}_t+ \qmexp{n_0}_t^2}\\
    &\le \sum_{x\in\Lambda\setminus\Set{0}} \abs{J_{x,0}}\abs{x}^{\beta+1}\del{\qmexp{n_0}_t+\qmexp{n_0}_t^2\del{\sum_{y\in\Lambda}\mathbbm 1_{|y|\leq R }+\sum_{y\in\Lambda}\mathbbm 1_{|y-x|\leq R }}}\\
    &\le 3 \abs{B_R}\qmexp{n_0}_t^2\sum_{x\in\Lambda\setminus\Set{0}} \abs{J_{x,0}}\abs{x}^{\beta+1}\\
    &\le 3\kappa^{(\beta+1)}\abs{B_R}\qmexp{n_0}_t^2.
\end{align}
All in all, we obtain
\begin{align}\label{II good rem}
    \qmexp{\mathrm{(II)}}_t
    &\leq \frac{\kappa}{s} \qmexp{\bA'_{ts}}_t+\frac{C}{s^2}\qmexp{\tilde\bA'_{ts}}_t+C\frac{\kappa^{(\beta+1)}}{s^{\beta+1}}\abs{B_R}\qmexp{n_0^2}_t.
\end{align}
Now consider again \eqref{I + II}. We can control its r.h.s. applying \eqref{I astlo no rem} and \eqref{II good rem}.
\begin{align}\label{control comm }
    \sbr{iH,\bA_{ts}}&\le \del{\abs{I}+\abs{II}}\notag\\
    &\le \frac{2\kappa}{s}\qmexp{\mathbb A'_{ts}}_t+\frac{C}{s^{2}} \qmexp{\tilde{\A}'_{ts}}_t+\frac{C}{s^{\beta+1}}\abs{B_R}\kappa^{(\beta+1)}\qmexp{n_0^2}_t\notag\\
    &\le \frac{2\kappa}{s}\qmexp{\mathbb A'_{ts}}_t+\frac{C}{s^{2}} \qmexp{\tilde{\A}'_{ts}}_t+\frac{C}{s^{\beta+1}}R^d\qmexp{n_0^2}_t.
\end{align}
  Notice that the second-order term in \eqref{control comm } should be dropped for $\beta=1$. This concludes the proof.

\end{proof}

\begin{proof}[Proof of \propref{bootstrapping}]

The idea for obtaining the desired estimate is to combine \eqref{contr der w rem} and \eqref{eq:ftc} with a bootstrapping strategy. 
Let us start considering \eqref{eq:ftc} and bounding the derivative appearing on the r.h.s. using \eqref{contr der w rem}. 
    \begin{align}\label{ftc + contr der}
        \qmexp{n_0 N_{B_r}}_t&\le R^d\lambda^2+\int_0^t \frac{d}{d\tau}\qmexp{\A_{\tau s}}_\tau d\tau\\
        &\le R^d\lambda^2+\frac{2\kappa-v}{s}\int_0^t\qmexp{\A'_{\tau s}}_\tau d\tau+C\frac{1}{s^2}\int_0^t\langle  \tilde \A'_{\tau s}\rangle_\tau d\tau+\frac{Ct}{s^{\beta+1}}R^d\sup_{vt\le R-r}\qmexp{n_0^2}_t.\notag
    \end{align}
Recall the remainder does not depend on $\chi$. 
 We drop the term on the l.h.s. of \eqref{ftc + contr der}, and we focus on the first integrated term on the r.h.s..  Since its coefficient is negative, due to the fact $v>2\kappa$, we can move it to the l.h.s.,
\begin{align}\label{beg boot}
    \frac{1}{s}\int_0^t\qmexp{\A'_{\tau s}}_\tau d\tau\le  CR^d\lambda^2+\frac{C}{s^2}\int_0^t\langle  \tilde \A'_{\tau s}\rangle_\tau d\tau+\frac{Ct}{s^{\beta+1}}R^d\sup_{vt\le R-r}\qmexp{n_0^2}_t.
\end{align}
Notice that for $\beta=1$ the second term on the r.h.s. of \eqref{beg boot} should be dropped.
We observe that inequality \eqref{beg boot} holds for any $\chi\in\cC$, so we can apply it, for $\tilde\chi$ and $\beta-1$, to the integrated term on the r.h.s. of \eqref{beg boot} itself.
\begin{align}
    \frac{1}{s}\int_0^t\qmexp{\A'_{\tau s}}_\tau d\tau&\le  CR^d\lambda^2\del{1+\frac{C}{s}}+   \frac{C}{s^3}\int_0^t\qmexp{\bar \A'_{t s}}_\tau d\tau+\frac{2Ct}{s^{\beta+1}}R^d\sup_{vt\le R-r}\qmexp{n_0^2}_t,\notag
\end{align}
for some new $\bar{\chi}\in\cC$.
We obtained an inequality very similar in structure to \eqref{beg boot}, with the advantage of having an additional factor $s^{-1}$ in front of the integrated ASTLO. We iterate this procedure, applying \eqref{beg boot} for $\beta-2,\,\beta-3,\dots$, and we obtain
\begin{align}\label{contr int astlo}
    \frac{1}{s}\int_0^t\qmexp{\A'_{\tau s}}_\tau d\tau&\le  CR^d\lambda^2\del{1+C\sum_{k=1}^{\beta-1}{s^{-k}}}+ \frac{Ct}{s^{\beta+1}}R^d\sup_{vt\le R-r}\qmexp{n_0^2}_t\notag\\
    &\le  CR^d\lambda^2+  \frac{ Ct}{s^{\beta+1}}R^d\sup_{vt\le R-r}\qmexp{n_0^2}_t.\notag
\end{align}
Consider again \eqref{ftc + contr der} and drop the negative term on the r.h.s.. We apply \eqref{contr int astlo} to control the leftover integrated ASTLO, and we obtain
\begin{align}
     \qmexp{n_0 N_{B_r}}_t\le R^d\lambda^2\del{1+Cs^{-1}}+\frac{C t}{s^{\beta+1}}R^d\sup_{vt\le R-r}\qmexp{n_0^2}_t.\notag
\end{align}
This finishes the proof.
\end{proof}

\section*{Acknowledgments}

TK is supported by JST PRESTO (Grant No.
JPMJPR2116), Exploratory Research for Advanced Technology (Grant No. JPMJER2302), and JSPS Grants-in-
Aid for Scientific Research (No. JP23H01099, JP24H00071), Japan.
The research of ML is supported by the DFG through the grant TRR 352 – Project-ID 470903074 and by
the European Union (ERC Starting Grant MathQuantProp, Grant Agreement 101163620)\footnote{Views and opinions expressed are however those of the authors only and do not necessarily reflect those of the European Union or the
European Research Council Executive Agency. Neither the European Union nor the granting authority can be held responsible for
them.}. 
The research of CR is supported by the DFG through the grant
TRR 352 – Project-ID 470903074.

\bibliography{bibfile}

\begin{bibdiv}
\begin{biblist}

\bib{APSS}{article}{
      author={Arbunich, J.},
      author={Pusateri, F.},
      author={Sigal, I.~M.},
      author={Soffer, A.},
       title={Maximal speed of quantum propagation},
        date={2021},
     journal={Lett. Math. Phys.},
      volume={111},
      number={62},
}

\bib{Bony2012-pz}{article}{
      author={Bony, Jean-Fran{\c c}ois},
      author={Faupin, J{\'e}r{\'e}my},
      author={Sigal, Israel~Michael},
       title={Maximal velocity of photons in non-relativistic {QED}},
    language={en},
        date={2012-12},
     journal={Adv. Math. (N. Y.)},
      volume={231},
      number={5},
       pages={3054\ndash 3078},
}

\bib{bravyi2006lieb}{article}{
      author={Bravyi, Sergey},
      author={Hastings, Matthew~B},
      author={Verstraete, Frank},
       title={{Lieb-Robinson bounds and the generation of correlations and topological quantum order}},
        date={2006},
     journal={Physical Review Letters},
      volume={97},
      number={5},
       pages={050401},
}

\bib{CGS}{article}{
      author={Cevolani, Lorenzo},
      author={Carleo, Giuseppe},
      author={Sanchez-Palencia, Laurent},
       title={Protected quasilocality in quantum systems with long-range interactions},
        date={2015},
     journal={Phys. Rev. A},
      volume={92},
      number={4},
       pages={041603},
}

\bib{chen2023speed}{article}{
      author={Chen, Chi-Fang},
      author={Lucas, Andrew},
      author={Yin, Chao},
       title={{Speed limits and locality in many-body quantum dynamics}},
        date={2023},
     journal={Reports on Progress in Physics},
}

\bib{MR4419446}{article}{
      author={Faupin, J\'{e}r\'{e}my},
      author={Lemm, Marius},
      author={Sigal, Israel~Michael},
       title={Maximal speed for macroscopic particle transport in the {B}ose-{H}ubbard model},
        date={2022},
        ISSN={0031-9007,1079-7114},
     journal={Phys. Rev. Lett.},
      volume={128},
      number={15},
       pages={Paper No. 150602, 6},
         url={https://doi.org/10.1103/physrevlett.128.150602},
}

\bib{MR4470244}{article}{
      author={Faupin, J\'{e}r\'{e}my},
      author={Lemm, Marius},
      author={Sigal, Israel~Michael},
       title={On {L}ieb-{R}obinson bounds for the {B}ose-{H}ubbard model},
        date={2022},
        ISSN={0010-3616,1432-0916},
     journal={Commun. Math. Phys.},
      volume={394},
      number={3},
       pages={1011\ndash 1037},
         url={https://doi.org/10.1007/s00220-022-04416-8},
}

\bib{faupin2025macroscopic}{article}{
      author={Faupin, J{\'e}r{\'e}my},
      author={Lemm, Marius},
      author={Sigal, Israel~Michael},
      author={Zhang, Jingxuan},
       title={Macroscopic suppression of supersonic quantum transport},
        date={2025},
     journal={Physical Review Letters},
      volume={135},
      number={16},
       pages={160405},
}

\bib{Fossetal}{article}{
      author={Foss-Feig, M.},
      author={Gong, Z.-X.},
      author={Clark, C.~W.},
      author={Gorshkov, A.~V.},
       title={Nearly-linear light cones in long-range interacting quantum systems},
        date={2015},
     journal={Phys. Rev. Lett.},
      volume={114},
       pages={157201},
}

\bib{gogolin2016equilibration}{article}{
      author={Gogolin, Christian},
      author={Eisert, Jens},
       title={{Equilibration, thermalisation, and the emergence of statistical mechanics in closed quantum systems}},
        date={2016},
     journal={Reports on Progress in Physics},
      volume={79},
      number={5},
       pages={056001},
}

\bib{Hastings2004-ph}{article}{
      author={Hastings, M~B},
       title={{Lieb-Schultz-Mattis} in higher dimensions},
        date={2004-03},
     journal={Phys. Rev. B Condens. Matter Mater. Phys.},
      volume={69},
      number={10},
}

\bib{hastings2007area}{article}{
      author={Hastings, Matthew~B},
       title={{An area law for one-dimensional quantum systems}},
        date={2007},
     journal={Journal of Statistical Mechanics: Theory and Experiment},
      volume={2007},
      number={08},
       pages={P08024},
}

\bib{hastings2006spectral}{article}{
      author={Hastings, Matthew~B},
      author={Koma, Tohru},
       title={{Spectral gap and exponential decay of correlations}},
        date={2006},
     journal={Communications in Mathematical Physics},
      volume={265},
       pages={781\ndash 804},
}

\bib{hastings2005quasiadiabatic}{article}{
      author={Hastings, Matthew~B},
      author={Wen, Xiao-Gang},
       title={{Quasiadiabatic continuation of quantum states: The stability of topological ground-state degeneracy and emergent gauge invariance}},
        date={2005},
     journal={Physical Review B},
      volume={72},
      number={4},
       pages={045141},
}

\bib{Herbst1995-ac}{article}{
      author={Herbst, Ira},
      author={Skibsted, Erik},
       title={Free channel fourier transform in the long-rangen-body problem},
    language={en},
        date={1995-12},
     journal={J. Anal. Math.},
      volume={65},
      number={1},
       pages={297\ndash 332},
}

\bib{KS_he}{article}{
      author={Kuwahara, T.},
      author={Saito, K.},
       title={Absence of fast scrambling in thermodynamically stable long-range interacting systems},
        date={2021},
     journal={Phys. Rev. Lett.},
      volume={126},
       pages={030604},
}

\bib{Kuwahara2024-nf}{article}{
      author={Kuwahara, Tomotaka},
      author={Lemm, Marius},
       title={Enhanced {Lieb-Robinson} bounds for a class of {Bose-Hubbard} type {Hamiltonians}},
        date={2024},
      eprint={2405.04672},
}

\bib{kuwahara2021lieb}{article}{
      author={Kuwahara, Tomotaka},
      author={Saito, Keiji},
       title={Lieb-robinson bound and almost-linear light cone in interacting boson systems},
        date={2021},
     journal={Physical review letters},
      volume={127},
      number={7},
       pages={070403},
}

\bib{KVS}{misc}{
      author={Kuwahara, Tomotaka},
      author={Vu, Tan~Van},
      author={Saito, Keiji},
       title={Optimal light cone and digital quantum simulation of interacting bosons},
   publisher={Preprint, arXiv 2206.14736},
        date={2022},
}

\bib{LRSZ}{article}{
      author={Lemm, Marius},
      author={Rubiliani, Carla},
      author={Sigal, Israel~Michael},
      author={Zhang, Jingxuan},
       title={Information propagation in long-range quantum many-body systems},
        date={2023},
     journal={Phys. Rev. A},
      volume={108},
       pages={L060401},
         url={https://link.aps.org/doi/10.1103/PhysRevA.108.L060401},
}

\bib{lemm2023microscopic}{misc}{
      author={Lemm, Marius},
      author={Rubiliani, Carla},
      author={Zhang, Jingxuan},
       title={On the microscopic propagation speed of long-range quantum many-body systems},
        date={2023},
}

\bib{Lemm2025-ov}{article}{
      author={Lemm, Marius},
      author={Rubiliani, Carla},
      author={Zhang, Jingxuan},
       title={On the quantum dynamics of long-ranged {Bose-Hubbard} hamiltonians},
        date={2025},
      eprint={2505.01786},
}

\bib{lemm2025enhanced}{article}{
      author={Lemm, Marius},
      author={Wessel, Tom},
       title={{Enhanced Lieb-Robinson bounds for commuting long-range interactions}},
        date={2025},
     journal={Journal of Mathematical Physics},
      volume={66},
      number={9},
}

\bib{lieb1972finite}{article}{
      author={Lieb, Elliott~H},
      author={Robinson, Derek~W},
       title={{The finite group velocity of quantum spin systems}},
        date={1972},
     journal={Communications in Mathematical Physics},
      volume={28},
      number={3},
       pages={251\ndash 257},
}

\bib{nachtergaele2006propagation}{article}{
      author={Nachtergaele, Bruno},
      author={Ogata, Yoshiko},
      author={Sims, Robert},
       title={{Propagation of correlations in quantum lattice systems}},
        date={2006},
     journal={Journal of Statistical Physics},
      volume={124},
       pages={1\ndash 13},
}

\bib{Nachtergaele2009-rl}{article}{
      author={Nachtergaele, Bruno},
      author={Raz, Hillel},
      author={Schlein, Benjamin},
      author={Sims, Robert},
       title={{Lieb-Robinson} bounds for harmonic and anharmonic lattice systems},
    language={en},
        date={2009-03},
     journal={Commun. Math. Phys.},
      volume={286},
      number={3},
       pages={1073\ndash 1098},
}

\bib{nachtergaele2006lieb}{article}{
      author={Nachtergaele, Bruno},
      author={Sims, Robert},
       title={{Lieb-Robinson bounds and the exponential clustering theorem}},
        date={2006},
     journal={Communications in Mathematical Physics},
      volume={265},
       pages={119\ndash 130},
}

\bib{SHOE}{article}{
      author={Schuch, N.},
      author={Harrison, S.~K.},
      author={Osborne, T.~J.},
      author={Eisert, J.},
       title={Information propagation for interacting-particle systems},
        date={2011},
     journal={Phys. Rev. A},
      volume={84},
       pages={032309},
}

\bib{SigSof}{misc}{
      author={Sigal, I.M.},
      author={Soffer, A.},
       title={Local decay and propagation estimates for time-dependent and time-independent hamiltonians},
         how={Preprint, Princeton Univ. (1988)},
         url={http://www.math.toronto.edu/sigal/publications/SigSofVelBnd.pdf},
}

\bib{SiZh}{misc}{
      author={Sigal, Israel~Michael},
      author={Zhang, Jingxuan},
       title={On propagation of information in quantum many-body systems},
   publisher={Preprint, arXiv 2212.14472},
        date={2022},
         url={https://arxiv.org/abs/2212.14472},
}

\bib{skibsted}{article}{
      author={Skibsted, Erik},
       title={Propagation estimates for n-body schroedinger operators},
    language={English},
        date={1991-11},
        ISSN={0010-3616},
     journal={Communications in Mathematical Physics},
      volume={142},
      number={1},
       pages={67\ndash 98},
}

\bib{TranEtAl2}{article}{
      author={Tran, M.~C.},
      author={Guo, A.~Y.},
      author={Su, Y.},
      author={Garrison, J.},
      author={James, R.},
      author={Eldredge, Z.},
      author={Foss-Feig, M.},
      author={Childs, A.~M.},
      author={Gorshkov, A.~V.},
       title={Locality and digital quantum simulation of power-law interaction},
        date={2019},
     journal={Phys. Rev. X},
      volume={9},
       pages={031006},
}

\bib{TranEtal4}{article}{
      author={Tran, Minh~C.},
      author={Chen, Chi-Fang},
      author={Ehrenberg, Adam},
      author={Guo, Andrew~Y.},
      author={Deshpande, Abhinav},
      author={Hong, Yifan},
      author={Gong, Zhe-Xuan},
      author={Gorshkov, Alexey~V},
      author={Lucas, Andrew},
       title={Hierarchy of linear light cones with long-range interactions},
        date={2020},
     journal={Physical Review X},
      volume={10},
      number={3},
       pages={031009},
}

\bib{TranEtAl1}{article}{
      author={Tran, Minh~C.},
      author={Guo, Andrew~Y.},
      author={Baldwin, Christopher~L.},
      author={Ehrenberg, Adam},
      author={Gorshkov, Alexey~V.},
      author={Lucas, Andrew},
       title={{L}ieb-{R}obinson light cone for power-law interactions},
        date={2021},
     journal={Phys. Rev. Lett.},
      volume={127},
       pages={160401},
         url={https://link.aps.org/doi/10.1103/PhysRevLett.127.160401},
}

\bib{van2024optimal}{article}{
      author={Van~Vu, Tan},
      author={Kuwahara, Tomotaka},
      author={Saito, Keiji},
       title={Optimal light cone for macroscopic particle transport in long-range systems: A quantum speed limit approach},
        date={2024},
     journal={Quantum},
      volume={8},
       pages={1483},
}

\bib{wang2020tightening}{article}{
      author={Wang, Zhiyuan},
      author={Hazzard, Kaden~RA},
       title={Tightening the lieb-robinson bound in locally interacting systems},
        date={2020},
     journal={PRX Quantum},
      volume={1},
      number={1},
       pages={010303},
}

\bib{YL}{article}{
      author={Yin, Chao},
      author={Lucas, Andrew},
       title={Finite speed of quantum information in models of interacting bosons at finite density},
        date={2022},
     journal={Phys. Rev. X},
      volume={12},
       pages={021039},
         url={https://link.aps.org/doi/10.1103/PhysRevX.12.021039},
}

\end{biblist}
\end{bibdiv}
\end{document}